\documentclass[10pt,conference]{IEEEtran}
\usepackage{amssymb,color,amsmath} 
\usepackage[alwaysadjust]{paralist}
\usepackage{graphicx}
\usepackage{cite}
\usepackage[urlcolor=rltblue,colorlinks=true]{hyperref} %
\definecolor{rltblue}{rgb}{0,0,0.75}

\newcommand{\F}{\mathbf{F}}

\newcommand{\ord}{{\rm {ord}}}

\newtheorem{theorem}{\textbf{Theorem}}
\newtheorem{lemma}[theorem]{\textbf{Lemma}}
\newtheorem{defn}[theorem]{\textbf{Definition}}

\newcommand{\ket}[1]{|#1\rangle}

\newcommand{\nix}[1]{}

\begin{document}
\title{A Class of Quantum LDPC Codes Constructed From Finite Geometries}
\author{
\authorblockN{Salah A. Aly\\}
\authorblockA{Department  of Computer Science, Texas A\&M University\\
College Station, TX 77843, USA \\
Email: salah@cs.tamu.edu }
 }  \maketitle

\begin{abstract}
Low-density parity check (LDPC) codes are a significant class of
classical codes with many applications. Several good LDPC codes have
been constructed using random, algebraic, and finite geometries
approaches, with containing cycles of length at least six in their
Tanner graphs. However, it is impossible to design a self-orthogonal
parity check matrix of an LDPC code without introducing cycles of
length four.

In this paper, a new class of quantum LDPC codes based on lines and
points of finite geometries is constructed. The parity check
matrices of these codes are adapted to be self-orthogonal with
containing only one cycle of length four in each pair of two rows.
Also, the column and row weights, and bounds on the minimum distance
of these codes are given. As a consequence, these codes can be
encoded using shift-register encoding algorithms and can be decoded
using iterative decoding algorithms over various quantum
depolarizing channels.
\end{abstract}
\nix{
}

\section{Introduction}\label{sec:intro}
Low density parity check (LDPC) codes are a
capacity-approaching~(\emph{Shannon limit}) class of codes that were
first described in a seminal work by Gallager~\cite{gallager62}. In
Tanner~\cite{tanner81}, LDPC codes were rediscovered and presented
in a graphical interpretation~(\emph{codes over graphs}).  Iterative
decoding of LDPC and turbo codes highlighted the importance of these
classes of codes for communication and storage channels.
Furthermore, they have been used extensively in many
applications~\cite{macKay98,lin04,liva06}.

 There have been several notable attempts to construct regular and irregular
good LDPC codes using algebraic combinatorics and random
constructions, see~\cite{song06,liva06}, and references
therein. Liva~\emph{et al.}~\cite{liva06} presented a survey of the
previous work done on algebraic constructions of LDPC codes based on
finite geometries, elements of finite fields, and RS codes.
Furthermore, a good construction of LDPC codes should have a girth
of the Tanner graph, of at least six~\cite{liva06,lin04}.

Quantum information is sensitive to noise and needs error
correction, control, and recovery strategies. Quantum block and
convolutional codes are means to protect quantum information against
noise and decoherence. A well-known class of quantum codes is called
stabilize codes, in which it can be easily constructed using
self-orthogonal (or dual-containing) classical codes,
see~\cite{calderbank98,aly07a,ketkar06} and references therein.
 Recently, subsystem codes combine the features of decoherence free subspaces,
noiseless subsystems, and quantum error-correcting codes,
see~\cite{aly06c,bacon06,kribs05b,lidar98} and references
therein.

Quantum block LDPC codes have been proposed
in~\cite{postol01,macKay04}. MacKay \emph{et al.} in~\cite{macKay04}
constructed sparse graph quantum LDPC codes based on cyclic matrices
and using a computer search. Recently, Camera \emph{el al.} derived
quantum LDPC codes in an analytical method~\cite{camara05}. Hagiwara
and Imai constructed quasi-cyclic (QC) LDPC codes and derived  a
family of quantum QC LDPC codes from a nested pair of classical
codes~\cite{hagiwara07}.

 In this paper, we construct LDPC codes based on finite
geometry. We show that the constructed LDPC codes have quasi-cyclic
structure and their parity check matrices can be adapted to satisfy
the  self-orthogonal (or dual-containing) conditions. The
motivations for this work are that \begin{inparaenum}[(i)] \item
LDPC codes constructed from finite geometries can be encoded using
linear shift-registers. The column weights remain fixed with the
increase in  number of rows and length of the code.
\item The adapted parity check matrix has exactly one cycle with length four between any two rows and many cycles with length of at least six. \item A class of quantum LDPC codes is constructed that can be decoded using
known iterative decoding algorithms over quantum depolarizing
channels; some of these algorithms are stated in~\cite{poulin07}.
\end{inparaenum}

 \emph{Notation:} Let $q$ be a prime power
$p$ and $\F_q$ be a finite field with $q$ elements. Any two binary
vectors $\textbf{v}=(v_1,v_2,\ldots,v_n)$ and
$\textbf{u}=(u_1,u_2,\ldots,u_n)$ are orthogonal if their inner
product vanishes, i.e., $\sum_{i=1}^n v_iu_i \mod 2=0$. Let
$\textbf{H}$ be a parity check matrix defined over $\F_2$, then
\textbf{H} is self-orthogonal if the inner product between any two
arbitrary rows of \textbf{H} vanishes.

\section{LDPC Code Constructions and Finite Geometries}
\subsection{LDPC Codes}
\begin{defn} An $(\rho, \lambda)$ regular LDPC code is defined by a sparse
binary parity check matrix $\textbf{H}$ satisfying the following
properties.
\begin{compactenum}[i)]
\item $\rho$ is the number of one's in a column.
\item $\lambda$ is the number of one's in  a row.
\item Any two rows have at most one nonzero element in common. The code does not have cycles of length four in its Tanner graph.
\item $\rho$ and $\lambda$ are small in comparison to the number of rows and length of
the code. In addition, rows of the matrix \textbf{H} are not
necessarily linearly independent.
\end{compactenum}
\end{defn}

The third condition guarantees that  iterative decoding algorithms
such as sum-product or message passing perform well over
communication channels. In general it is hard to design regular LDPC
satisfying the above conditions, see
\cite{song06,liva06,lin04} and references therein.
\subsection{Finite Geometry}
Finite geometries can be classified into Euclidean and projective
geometry over finite fields.  Finite geometries codes are an
important class of cyclic and quasi-cyclic codes because their
encoder algorithms can be implemented using linear feedback shift
registers and their decoder algorithms can be implemented using
various decoding algorithms such as majority logic (MLG),
sum-product (SPA), and weighted BF, see~\cite{kou01,liva06,lin04}.

\begin{defn}
A finite geometry with  a set of $n$ points $\{p_1,p_2,\dots,p_n\}$,
a set of $l$ lines $\{L_1,L_2,\ldots,L_l\}$ and an integer pair
$(\lambda,\rho)$ is defined as follows:
\begin{compactenum}[i)]
\item Every line $L_i$ passes through $\rho$ points.
\item Every point $p_i$ lies in $\lambda$ lines, i.e., every point
$p_i$ is intersected by $\lambda$ lines.
\item Any two points $p_1$ and $p_j$ can define one and only one line $L_k$ in between.
\item Any two lines $L_i$  and $L_j$ either intersect at only one point $p_i$ or they are parallel.
\end{compactenum}
\end{defn}

Therefore, we can form a binary matrix $\textbf{H}=[h_{i,j}]$ of
size $l \times n$ over $\F_2$. The rows and columns of \textbf{H}
correspond the $l$ lines and $n$ points in the Euclidean geometry,
respectively. If the i\emph{th} line $L_i$ passes through the point
$p_i$ then $h_{i,j}=1$, and otherwise $h_{i,j}=0$
\begin{figure}[h]
  \includegraphics[scale=0.7]{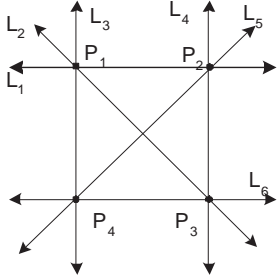}
 \centering
  \caption{Euclidean geometry with points $n=4$ and lines $l=6$}\label{ldpc1}
\end{figure}
Fig.~\ref{ldpc1} shows an example of Euclidean geometry with $n=4$,
$l=6$, $\lambda=3$, and $\rho=2$. We can construct the incidence
matrix $\textbf{H}$ based on this geometry where every point and
line correspond to a column and row, respectively.  For $\rho << l$
and $\lambda <<n$, The matrix $\textbf{H}$ is a sparse low density
parity check matrix. In this example, the matrix $\textbf{H}_{EG-I}$
is given by

\begin{eqnarray}
\textbf{H}_{EG-I}=\left( \begin{array}{cccc}1&1&0&0 \\1&0&1&0 \\ 1&0&0&1 \\ 0&1&1&0\\
0&1&0&1\\ 0&0&1&1 \end{array} \right)
\end{eqnarray}

We call the Euclidean geometry defined in this type as a
$\textbf{Type-I EG}$. The Tanner graph of \textbf{Type-I EG}  is a
regular bipartite graph with $n$ code variable vertices and $l$
check-sum vertices. Also, each variable bit vertex has degree
$\lambda$ and each check-sum has degree $\rho$.

If we can take the transpose of this matrix $\textbf{H}_{EG-I}$,
then we can also define a $(\rho,\lambda)$ LDPC code with length $l$
and minimum distance is at least $\rho+1$. The codes defined in this
type are called LDPC codes based on $\textbf{Type-II EG}$. In this
type, any two rows intersect at exactly one position.

\begin{figure}[t]
  \includegraphics[scale=0.7]{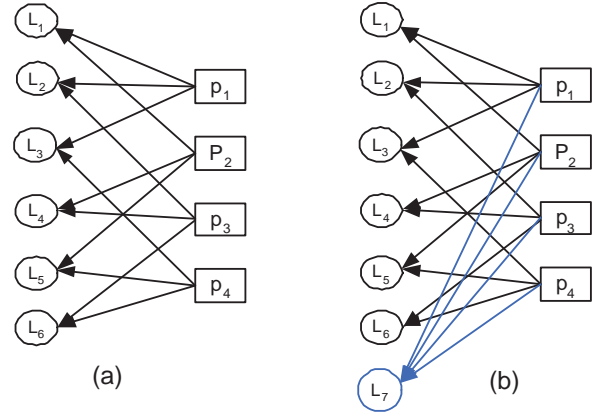}
 \centering
  \caption{(a) Euclidean geometry with $n=4$ points (check-sum) and $l=6$ lines (code variables) (b) This Tanner graph corresponds to a self-orthogonal parity check
matrix.}\label{ldpc2}
\end{figure}

\subsection{Adapting the Matrix $\textbf{H}_{EG-II}$ to be Self-orthogonal}
Let $\textbf{H}_{EG-II}$  be  a parity check matrix of a regular
LDPC code constructed based on \textbf{Type-II EG} Euclidean
geometry. We can construct a self-orthogonal matrix
$\textbf{H}_{EG-II}^{orth}$ from $\textbf{H}_{EG-II}$ in two cases.

\textbf{Case 1.} If the number of one's in a row is odd and any two
rows intersect at exactly one position, i.e., any line connects two
points.  As shown in Fig.~\ref{ldpc2}, the Tanner graph corresponds
to a self-orthogonal parity check matrix $\textbf{H}_{EG-II}^{orth}$
if and only if every check-sum has even degree and any any two
check-sum nodes meet at even code variable nodes. This condition is
the same as every row in the parity check matrix
$\textbf{H}_{EG-II}^{orth}$ has an even weight and any two rows
overlap in even nonzero positions.

\begin{small}
\begin{eqnarray}
\textbf{H}_{EG-II}^{orth}&=&\Big(\begin{array}{c|c}
\textbf{H}^T&\textbf{1}
\end{array} \Big)\end{eqnarray}
\end{small}
The vector $\textbf{1}$ of length $n$ is added as the last column in
$\textbf{H}_{EG-II}^{orth}$.

 \textbf{Case 2.} Assume the number of
one's in a line is even and any two rows intersect at exactly one
position. We can construct a self-orthogonal parity check matrix
$\textbf{H}_{EG-II}^{orth}$ as follows. We add the vector
$\textbf{1}$ along with the identity matrix \textbf{I} of size
$n\times n$. We guarantee that  any two rows of the matrix
$\textbf{H}_{EG-II}^{orth}$ intersect at two nonzero positions and
every row has an even weight.
\begin{eqnarray}
\textbf{H}_{EG-II}^{orth}&=&\left(\begin{array}{c|c|c}
\textbf{H}^T&\textbf{1}&\textbf{I}\end{array} \right).
\end{eqnarray}

\subsection{Characteristic Vectors and Matrices}

Let $n$ be a positive integer such that $n=q^m-1$, where
$m=\ord_n(q)$ is the multiplicative order of $q$ modulo $n$. Let
$\alpha$ denote a fixed primitive element of~$\F_{q^m}$. Define a
map $\textbf{z}$ from $\F_{q^m}^*$ to $\F_2^n$ such that all entries
of $\textbf{z}(\alpha^i)$ are equal to 0 except at position $i$,
where it is equal to 1. For example,
$\textbf{z}(\alpha^2)=(0,1,0,\ldots,0)$.  We call
$\textbf{z}(\alpha^k)$ the location (or characteristic) vector of
$\alpha^k$. We can define the location vector
$\textbf{z}(\alpha^{i+j+1})$ as the right cyclic shift of the
location vector $\textbf{z}(\alpha^{i+j})$, for $0 \leq j \leq n-1$,
and the power is taken module $n$. The location vector can be
extended to two or more nonzero positions. for example, the location
vector of $\alpha^2$, $\alpha^3$ and $\alpha^5$ is given by
$\textbf{z}(\alpha^2,\alpha^3,\alpha^5)=(0,1,1,0,1,0,\ldots,0)$.

\begin{defn}\label{def:Amatrix}We can define a map $A$ that associates to an element $\F_{q^m}^*$ a circulant
matrix in $\F_2^{n\times n}$ by
\begin{eqnarray}\label{label:mapA} A(\alpha^i)=\left ( \begin{array}{ccc}
\textbf{z}(\alpha^i) \\  \textbf{z}(\alpha^{i+1})
\\ \vdots \\  \textbf{z}(\alpha^{i+n-1})
\end{array} \right).
\end{eqnarray}
By construction, $A(\alpha^k)$ contains a 1 in every row and column.
\end{defn}

We will use the map $A$ to associate to a parity check matrix
$H=(h_{ij})$ in $(\F_{q^m}^*)$ the (larger and binary) parity check
matrix $\textbf{H}=(A(h_{ij}))$ in $\F_2^{n \times n}$. The matrices
$A(h_{ij})$$'s$ are $n \times n$ circulant permutation matrices
based on some primitive elements $h_{ij}$ as shown in
Definition~\ref{def:Amatrix}.

\section{Constructing Self-Orthogonal Cyclic LDPC Codes from Euclidean
Geometry}\label{sec:LDPCcodes}
In this section we construct self-orthogonal algebraic Low Density
Parity Check (LDPC) codes based on finite   geometries. Particulary,
there are two important classes of finite geometries: Euclidean and
projective geometry.

\subsection{Euclidean Geometry $EG(m,q)$}
We construct regular LDPC codes based on lines and points of
Euclidean geometry. The class we derive has a cyclic structure, so
it is called cyclic LDPC codes. Cyclic LDPC codes can be defined by
a sparse parity check matrix or by a generator polynomial and can be
encoded using shift-register. Furthermore, they can be decoded using
well-known iterative decoding algorithms~\cite{lin04,liva06}.

Let $q$ be power of a prime $p$, i.e. $q=p^s$ for some integer
$s\geq 2$. Let $EG(m,q)$ be the m-dimensional Euclidean geometry
over $\F_q$ for some integer $m \geq 2$. It consists of $p^{ms}=q^m$
points and every point is represented by an m-tuple,
see~\cite{kou01}. A line in $EG(m,q)$ can be described by a
$1$-dimensional subspace of the vector space of all $m$-tuples over
$\F_q$ or a coset of it. The number of lines in $EG(m,q)$ is given
by \begin{eqnarray}(q^{m-1})(q^{m}-1)/(q-1),\end{eqnarray} and each
line passes through $q$ points. Every line has $q^{(m-1)}-1$ lines
parallel to it. Also, for any point in $EG(m,q)$, there are
\begin{eqnarray}(q^{m}-1)/(q-1),\end{eqnarray} lines intersect at this point. Two lines can
intersect at only one point or they are parallel.

Let $\F_{q^m}$ be the extension field of $\F_q$. We can represent
each element in $\F_{q^m}$ as an $m$-tuple over $\F_q$. Every
element in the finite field $\F_{q^m}$ can be looked as  a point in
the Euclidean geometry $EG(m,q)$, henceforth $\F_{q^m}$ can be
regarded as the Euclidean geometry $EG(m,q)$.

Let $\alpha$ be a primitive element of $\F_{q^m}$. $q^m$ points of
$EG(m,q)$ can be represented by elements of the set $\{
0,1,\alpha,\alpha^2,\ldots,\alpha^{q^m-2} \}$. We can also define a
line $L$ as the set of points of the form $\{ \textbf{a}+\gamma
\textbf{ b} \mid \gamma \in \F_q \}$, where $\textbf{a}$ and
$\textbf{b}$ are linearly independent over $\F_{q}$. For a given
point $\textbf{a}$, there are $(q^m-1)/(q-1)$ lines in $EG(m,q)$
that intersect at $\textbf{a}$.

\textbf{Type-I EG.} Let $n=q^m-1$ be the number of points excluding
the original point $\textbf{0}$ in $EG(m,q)$. Assume $L$ be a line
not passing through $\textbf{0}$. We can define the binary vector
\begin{eqnarray}\textbf{v}_L=(v_1,v_1,\ldots,v_{n}),\end{eqnarray}
 where $v_i=1$ if the
point $\alpha^i$ lies in a line $L$.  The vector $\textbf{v}_L$ is
called the incidence vector of $L$. Elements of the vector
$\textbf{v}_L$ correspond to the elements
$1,\alpha,\alpha^2,\ldots,\alpha^{n-1}$. $\alpha L$ is also a line
in $EG(m,q)$, therefore $\alpha \textbf{v}_L$ is a right
cyclic-shift of the vector $\textbf{v}_L$. Clearly, the lines
$L,\alpha L,\ldots,\alpha^{n-1}L$ are all different. But, they  may
not be linearly independent.

Consider the vectors $L_i,\alpha L_i,\ldots,\alpha^{n-1} L_i$. We
can construct an $n \times n$ matrix $H_i$ in the form

\begin{eqnarray} H_i=\left ( \begin{array}{cccccc} \textbf{v}_{L_i} \\ \alpha
\textbf{v}_{L_i}\\ \vdots \\ \alpha^{n-1}\textbf{v}_{L_i}
\end{array} \right)
\end{eqnarray}
Clearly, $H_i$ is a circulant matrix with column and row weights
equals to $q$, the number of points that lie in a  line $\alpha^j
L_i$, for $0 \leq j \leq n-1$. $H_i$ has size of $n \times n$. The
total number of lines   in $EG(m,q)$ that do not pass through the
origin $\textbf{0}$ are given by
\begin{eqnarray}(q^{m-1}-1)(q^m-1)/(q-1)\end{eqnarray}
They can be partitioned into $(q^{m-1}-1)/(q-1)$ cyclic classes,
see~\cite{liva06}. Every class $\mathcal{H}_i$ can be defined by an
incidence vector $L_i$ as $\{L_i,\alpha L_i,\alpha^2
L_i,\ldots,\alpha^{n-1} L_i\}$ for $1 \leq i \leq
(q^{m-1}-1)/(q-1)$. Let $1 \leq \ell \leq (q^{m-1}-1)/(q-1)$, then
$\mathcal{H}_{EG,\ell}$ is defined as

\begin{eqnarray} \mathcal{H}_{EG,\ell}= \Big[ \begin{array}{cccccc} \mathcal{H}_1 & \mathcal{H}_2 &
\ldots & \mathcal{H}_\ell
\end{array} \Big]^T.
\end{eqnarray}

For each cyclic class $\mathcal{H}_i$, we can form the matrix
$\mathbf{H}_i$ over $\F_2$ of size $n \times n$. Therefore,
$\mathbf{H}_i$ is a circulant binary matrix of row and column
weights of q.

If we assume that there are $1\leq \ell \leq (q^{m-1}-1)/(q-1)$
incidence lines in $EG(m,q)$ not passing through the origin, then we
can form the binary matrix

\begin{eqnarray} \textbf{H}_{EG,\ell}=\Big [ \begin{array}{cccccc} \textbf{H}_1 &
\textbf{H}_2&
\ldots& \textbf{H}_\ell
\end{array} \Big]^T.
\end{eqnarray}

The matrix $\textbf{H}_{EG,\ell}$ consists of a $\ell$ sub-matrices
$\textbf{H}_i$ of size $n \times n$ and it has column and row
weights $\ell q$ and $q$, respectively. The null space of the matrix
$\textbf{H}_{EG,\ell}$ gives a cyclic EG-LDPC code of length
$n=q^m-1$ and minimum distance $\ell q+1$, whose Tanner graph has a
girth of at least six, see~\cite{song06,liva06}.

The Tanner graph of \textbf{Type-I EG}  is a regular bipartite graph
with $q^m-1$ code variable vertices and $l$ check-sum vertices.
Also, Each variable bit vertex has degree $\rho=q$ and each
check-sum has degree $\lambda=\ell q$.


\textbf{Type-II EG.} We can take the transpose of the parity check
matrix $\mathcal{H}_{(EG,\ell)}$ over $\F_{q^m}$ as defined  in
$\textbf{Type-I}$ to define a new parity check matrix with the
following properties, see~\cite{kou01}.
\begin{eqnarray} \mathcal{H}_{EG,\ell}^T= \Big[ \begin{array}{cccccc} \mathcal{H}_1^T & \mathcal{H}_2^T&
\ldots& \mathcal{H}_\ell^T
\end{array} \Big]
\end{eqnarray}
So, the matrix $\mathcal{H}_i^T$ is the transpose matrix of
$\mathcal{H}_i$. Consequently, we can define the binary matrix
$\textbf{H}_{EG,\ell}$

\begin{eqnarray} \textbf{H}_{EG,\ell}^T=\Big [ \begin{array}{cccccc} \textbf{H}_1^T & \textbf{H}_2^T&
\ldots& \textbf{H}_\ell^T
\end{array} \Big].
\end{eqnarray}

Let $\ell=(q^{m-1}-1)/(q-1)$, then the matrix
$\textbf{H}_{EG,\ell}^T$ has the following properties
\begin{compactenum}[i)]
\item
The total number of columns is given by $\ell n=
(q^{m-1}-1)(q^m-1)/(q-1)$.

\item Number of rows is given by $n=q^m-1$.
\item The rows of this matrix correspond to the nonorigin points of $EG(m,q)$
and the columns correspond to the lines in $EG(m,q)$ that do not
pass through the origin. \item $\lambda=\ell q=
q(q^{m-1}-1)/(q-1)=(q^m-1)/(q-1)-1$ is the row weight for $\ell =
(q^{m-1}-1)/(q-1)$. Also $\rho=q$ is the column weight.
\item Any two rows  in $\textbf{H}_{EG,\ell}^T$ have exactly one nonzero element in common. Also, any two
columns have at most one nonzero element in common.
\item The binary sub-matrix $\textbf{H}_i^T$ has size $(q^m-1)\times (q^m-1)$. Also, it can be
constructed using only one vector $\textbf{v}_L$ that will be
cyclically shifted $q^m-1$ times.
\end{compactenum}

\subsection{QC LDPC Codes} The matrix $\textbf{H}^{T}_{EG,\ell}$
defines a quasi-cyclic (QC) LDPC code of length $N=\ell
n=(q^{(m-1)}-1)(q^m-1)/(q-1)$ for $\ell=(q^{m-1}-1)/(q-1)$. The
matrix $\textbf{H}^{T}_{EG,\ell}$ has $ n=q^m-1$ rows that are not
necessarily independent. We can define a QC LDPC code over $\F_2$ as
the null-space of the matrix $\textbf{H}^{T}_{EG,\ell}$ of sparse
circulant sub-matrices of equal size. The matrix
$\textbf{H}^{T}_{EG,\ell}$ with parameters $(\rho, \lambda)$ has the
following properties.
\begin{compactenum}[i)]
\item $\rho=q$ is the weight of a column $c_i$. $\rho$ does not
depend on $m$, hence length of the code can be increased without
increasing the column weight.
\item $\lambda=\ell q$ is the weight of a row $r_i$. $\lambda$
depends on $m$, but the length of the code increases much faster
than $\lambda$.
\item Every  two columns intersect at most at one nonzero position. Every two rows  have exactly one and only one nonzero
position in common.
\end{compactenum}

 From this definition, the minimum distance of the LDPC code defined by the null-space of
$\textbf{H}^{T}_{EG,\ell}$ is at least $\rho+1$. This is because we
can add at least $\rho+1$ columns in the parity check matrix
$\textbf{H}^{T}_{EG,\ell}$ to obtain the zero column (rank of
$\textbf{H}^{T}_{EG,\ell}$ is at least $(\rho+1$)). Furthermore, the
girth of the Tanner graph for this matrix $\textbf{H}_i$ is at least
six, see~\cite{macKay98,song06}. This is a $(\rho,\lambda)$ QC LDPC
code based on $\textbf{Type-II EG}$.

\subsection{Self-orthogonal QC LDPC
Codes}\label{sec:LDPCcodesorthogonal}

We can define a self-orthogonal parity check matrix
$\textbf{H}^{orth}_{EG,\ell}$ from $\textbf{Type-II EG}$
construction as follows. The binary matrix
$\textbf{H}^{T}_{EG,\ell}$ of size $n \times \ell n$ for $1 \leq
\ell \leq (q^{m-1}-1)/(q-1)$  has row and column weights of
$\lambda=\ell q$ and $\rho=q$, respectively. Let \textbf{1} be the
column vector of size $(q^m-1) \times 1$ defined as
$\textbf{1}=(1,1,\ldots,1)^T$. If the weight of a row in
$\textbf{H}^{T}_{EG,\ell}$ is odd, then we can add the vector
\textbf{1} to form the matrix $\textbf{H}^{orth}_{EG,\ell}=\Big[
\textbf{H}_{EG,\ell}^T  \mid \textbf{1} \Big]$. Also, if the weight
of a row in $\textbf{H}^{T}_{EG,\ell}$ is even, then we can add the
vector \textbf{1} along with the identity matrix of size
$(q^m-1)\times (q^m-1)$ to form $\textbf{H}^{orth}_{EG,\ell}=\Big[
\textbf{H}_{EG,\ell}^T  \mid \textbf{1} \mid \textbf{I} \Big]$.
Therefore, we can prove that $\textbf{H}^{orth}_{EG,\ell}$ is
self-orthogonal as shown in the following Lemma.

\begin{lemma}\label{lem:Hself-orthogonal}
The parity check matrix $\textbf{H}^{orth}_{EG,\ell}$ defined as
\begin{eqnarray}\textbf{H}^{orth}_{EG,\ell} =\left\{
  \begin{array}{ll}
    \Big[\begin{array}{cccc|c}\textbf{H}_1^T&\textbf{H}_2^T&\dots&\textbf{H}_\ell^T
&\textbf{1}
\end{array}\Big], \mbox{for odd $\ell q$;}  \\ \hspace{0.2cm} \\
    \Big[\begin{array}{cccc|c|c}\textbf{H}_1^T&\textbf{H}_2^T&\dots&\textbf{H}_\ell^T
&\textbf{1}&\textbf{I}
\end{array}\Big],  \mbox{for even $\ell q$} \nonumber
  \end{array}
\right.
\end{eqnarray}
is self-orthogonal.
\end{lemma}

\begin{proof}
From the construction $\textbf{Type-II EG}$, any two different rows
intersect (overlap)  in exactly one nonzero position. If $\ell q$ is
odd, then adding the column vector \textbf{1} will result an even
overlap as well as rows of even weights. Therefore, the inner
product $\mod 2$ of any arbitrary rows vanishes. Also, if $\ell q$
is even, adding the columns $\Big[ \textbf{1} \mid \textbf{I} \Big]$
will produce row of even weights and the inner product $\mod 2$ of
any arbitrary rows vanishes.
\end{proof}

$\textbf{H}^{orth}_{EG,\ell}$ has size $n \times N$ for odd $\ell q$
where $n=q^m-1$, $N= n\ell +1$, and $1 \leq \ell \leq
(q^{(m-1)}-1)/(q-1)$. Also, it has length  $N=n(\ell+1)+1$ for even
$\ell q$.

%
The minimum distance of the LDPC codes constructed in this
type can be shown using the BCH bound as stated in the following
result.

\begin{lemma}
The minimum distance of an LDPC defined by the parity check matrix
$\textbf{H}^{orth}_{EG,\ell}$  is at least  $q+1$.
\end{lemma}


\section{Quantum LDPC Block Codes}\label{sec:QLDPCcodes}

In this section we derive a family of LDPC stabilizer codes derived
from LDPC codes based on finite geometries.  Let $P=\{I,X,Z, Y=iXZ\}$ be a
set of Pauli matrices defined as
\begin{eqnarray} I=\left( \begin{array}{cc} 1 &0 \\0&1 \end{array}\right),
X=\left(
\begin{array}{cc} 0 &1 \\1&0 \end{array}\right), Z=\left( \begin{array}{cc} 1
&0 \\0&-1 \end{array}\right) \end{eqnarray}
and the matrix $Y$ is the combination of the matrices $X$ bit-flip
and $Z$ phase-flip defined as $Y=iXZ=\left(
\begin{array}{cc} 0 &-i
\\i&0
\end{array}\right)$. Clearly, $$X^2=Z^2=Y^2=I.$$

A well-known method to construct quantum codes is by using the
stabilizer formalism, see for
example~\cite{aly08thesis,calderbank98,gottesman97,macKay04} and references
therein. Assume we have a stabilizer group $S$ generated by a set
$\{S_1,S_2,\ldots,S_{n-k}\}$ such that every two row operators
commute with each other. The error operator $S_j$ is a tensor
product of $n$ Pauli matrices. $$S_j=E_1\otimes E_2\otimes \ldots
\otimes E_n, \hspace{0.3cm} E_i \in P.$$ $S_j$ can be seen as a
binary vector of length $2n$~\cite{macKay04,calderbank98}.  A
quantum code $Q$ is defined as +1 joint eigenstates of the
stabilizer $S$. Therefore, a codeword state $\ket{\psi}$ belongs to
the code $Q$ if and only if
\begin{eqnarray}S_j\ket{\psi}=\ket{\psi} \mbox{ for all } S_j \in
S.\end{eqnarray}

\textbf{CSS Construction:} Let $\textbf{G}$ and $\textbf{H}$ be two
binary matrices define the classical code $C$ and dual code
$C^\perp$, respectively. The CSS construction assumes that the
stabilizer subgroup (matrix) can be written as
\begin{eqnarray}
\textbf{S} = \left( \begin{array}{c|c} \textbf{H} & \textbf{0}
\\\textbf{0} &\textbf{G}
\end{array}\right)
\end{eqnarray}
where $\textbf{H}$ and $\textbf{G}$ are $k \times n$ matrixes
satisfying $\textbf{HG}^T=\textbf{0}$. The quantum code with
stabilizer $\textbf{S}$ is able to encode $n-2k$ logical qubits into
$n$ physical qubits. If $\textbf{G}=\textbf{H}$, then  the
self-orthgonality or dual-containing condition becomes
$\textbf{HH}^T=\textbf{0}$. If $C$ is a code that has a parity check
matrix $\textbf{H}$, then $C^\perp \subseteq C$.

 \textbf{Constructing Dual-containing LDPC Codes:} Let us
construct the stabilizer matrix
\begin{eqnarray}
S_{stab} = \Big( \begin{array}{c|c} H_{X} & 0 \\0 &H_{Z}
\end{array}\Big).
\end{eqnarray}

The matrix $\textbf{H}_{EG,\ell}^{orth}$ is a binary self-orthogonal
matrix as shown in Section~\ref{sec:LDPCcodesorthogonal}. We replace
every nonzero element in $\textbf{H}_{EG,\ell}^{orth}$ by the Pauli
matrix $X$ to form the matrix $H_X$. Similarly, we replace every
nonzero element in $\textbf{H}_{EG,\ell}^{orth}$ by the Pauli matrix
$Z$ to form the matrix $H_Z$. Therefore the matrix $S_{stab}$ is
also self-orthogonal. We can assume that the matrix $H_X$ corrects
the bit-flip errors, while the matrix $H_Z$ corrects the phase-flip
errors, see~\cite{macKay04,aly08thesis}.
\begin{lemma}\label{def:qldpc}
A  quantum LDPC code $Q$ with rate $(n-2k)/n$ is a code whose
stabilizer matrix $S_{stab}$ of size $2k \times 2n$ has a pair
$(\rho,\lambda)$ where $\rho$ is the number of non-zero error
operators in a column and $\lambda$ is the number of non-zero error
operators in a row. Furthermore, $S_{stab}$ is constructed from a
binary self-orthogonal parity check matrix
$\textbf{H}_{EG,\ell}^{orth}$ of size $k \times n$.
\end{lemma}

Using Lemma~\ref{def:qldpc} and LDPC codes given by the parity check
matrix $\textbf{H}_{EG,\ell}^{orth}$ as shown in
Section~\ref{sec:LDPCcodesorthogonal}, we can derive a class of
quantum LDPC codes as stated in the following Lemma. 

\begin{theorem}\label{lem:qldpcparameters}
 Let $\textbf{H}_{EG,\ell}^{orth}$ be a parity check matrix of an LDPC code based on $EG(m,q)$, where $n=q^m-1$ and $1\leq \ell \leq (q^{m-1}-1)/(q-1)$. Then,
there exists a quantum LDPC code $Q$ with parameters $[[N,N-2n,\geq
q+1]]_2$ where $N=\ell n+1$ for odd $\ell q$ and $N=(\ell +1)n+1$
for even $\ell q$.
\end{theorem}

\begin{proof}
By Lemma~\ref{lem:Hself-orthogonal}, $\textbf{H}_{EG,\ell}^{orth}$
is self-orthogonal. Using Lemma~\ref{def:qldpc}, there exists a
quantum LDPC code with the given parameters.
\end{proof}

\section{Conclusion}
We constructed a class of quantum LDPC codes derived from finite
geometries. The constructed codes have high rates and their minimum
distances are bounded. They only have one cycle of length four between any two rows and many cycles of length of at least six.  A new class of
quantum LDPC codes based on projective geometries can be driven in a
similar way.

 \smallskip


\bibliographystyle{ieeetr}

\end{document}